\documentclass[journal,onecolumn]{IEEEtran}
\usepackage[T1]{fontenc}
\usepackage{amsmath,amssymb,amsfonts,amsthm,mathtools}
\usepackage{booktabs,array,graphicx,microtype}
\usepackage{physics}
\usepackage{cite}
\usepackage{aliascnt}
\usepackage[hidelinks]{hyperref}
\usepackage[capitalise,noabbrev]{cleveref}
\hypersetup{pdftitle={From Steane to A7: Quantum Codes from Invariant States},pdfauthor={Ian Teixeira}}
\newtheorem{theorem}{Theorem}
\newaliascnt{lemma}{theorem}\aliascntresetthe{lemma}
\newaliascnt{proposition}{theorem}\newtheorem{proposition}[proposition]{Proposition}\aliascntresetthe{proposition}
\newaliascnt{corollary}{theorem}\newtheorem{corollary}[corollary]{Corollary}\aliascntresetthe{corollary}
\theoremstyle{definition}
\newaliascnt{definition}{theorem}\newtheorem{definition}[definition]{Definition}\aliascntresetthe{definition}
\theoremstyle{remark}
\newaliascnt{remark}{theorem}\aliascntresetthe{remark}
\DeclareMathOperator{\End}{End}
\DeclareMathOperator{\supp}{supp}
\DeclareMathOperator{\wt}{wt}

\newcommand{\SU}{\mathrm{SU}}

\newcommand{\GL}{\mathrm{GL}}
\newcommand{\PSL}{\mathrm{PSL}}
\newcommand{\AGL}{\mathrm{AGL}}
\newcommand{\CC}{\mathbb C}
\newcommand{\RR}{\mathbb R}

\newcommand{\FF}{\mathbb F}
\newcommand{\ZZ}{\mathbb Z}
\newcommand{\G}{\mathsf G}
\newcommand{\tA}{\tilde{A}_7}
\newcommand{\cC}{\mathcal C}
\newcommand{\cA}{\mathcal A}
\newcommand{\cN}{\mathcal N}

\newcommand{\enc}{W}
\allowdisplaybreaks[1]
\title{From Steane to $A_7$: Quantum Codes from Invariant States}
\author{Ian Teixeira\thanks{I. Teixeira is with the Department of Mathematics, University of California San Diego, La Jolla, CA 92093 USA (e-mail: iteixeira@ucsd.edu).}}
\begin{document}
\maketitle
\begin{abstract}
We construct a two-parameter family of single-error-correcting seven-ququart codes with transversal $\tA$ symmetry, realizing the finite component of a two-qubit super-golden gate set. These $((7,4,3))_4$ codes encode two logical qubits and support non-Clifford operations by applying the same gate to each physical ququart.

Our method constructs invariant parent states rather than searching directly for a codespace. For $n$ systems of local dimension $q$, suppose the local symmetry is a unitary $t$-group and the permutation symmetry is $t$-transitive. Maximal mixing of all marginals on at most $t$ sites is then equivalent to linear equations in Schur--Weyl sector weights. Every feasible solution yields parents of pure distance at least $t+1$; puncturing one site produces a $q$-dimensional code of distance at least $t$ with the prescribed local symmetry acting transversally.

Steane and the nonstabilizer Fake Steane code reveal why the construction leaves continuous freedom. They share the same prescribed transversal Clifford and permutation symmetries, and their eight-qubit purifications lie on a circle: error correction fixes two invariant-sector norms but leaves the relative phase free. Keeping the same permutation group and replacing Clifford symmetry by $\tA\leq\SU(4)$ gives three fixed sector norms and two free relative phases, producing the new ququart family.
\end{abstract}

\section{Introduction}

A gate group may be attractive for efficient compilation even when no error-correcting code is known that realizes its gates transversally. Super-golden gate constructions make this question concrete: can their distinguished finite symmetry be implemented transversally, so that as many of the gates as possible are realized in a naturally fault tolerant manner? For the exceptional $A_7$ symmetry appearing in one of the most efficient two-qubit super-golden gate sets~\cite{dalal}, we answer this question with a two-parameter family of seven-ququart codes that correct arbitrary single-ququart errors.

In general, prescribing the symmetry does not by itself come anywhere close to enforcing error correction. Even after restricting the encoding to carry a chosen transversal action, the Knill--Laflamme conditions leave nonlinear equations~\cite{knilllaflamme,eastinknill}. However, previous work showed that, in certain very special cases, finite-design and representation-theoretic conditions can make portions of these equations automatic~\cite{exotic,twisted}. Here local unitary-design symmetry and a highly transitive permutation group reduce to linear equations in the squared norms of invariant components. Although the reduction here is not quite as strong, linear systems are still highly tractable and the new method applies to more symmetries and higher distances, so the tradeoff is strong, allowing us to construct the first quantum codes with $ \tA $ transversal gates (here $\tA\leq\SU(4)$ is the chosen matrix lift with projective image $A_7$; its generators and scalar convention are given in Appendix~\ref{app:a7-explicit}).

\begin{theorem}[Seven-ququart $\tA$ codes]\label{thm:main}
There is a two-parameter family of $((7,4,3))_4$ quantum codes and encoding isometries
\[
\enc_{\theta,\varphi}:\CC^4\longrightarrow(\CC^4)^{\otimes7}
\]
such that
\[
g^{\otimes7}\enc_{\theta,\varphi}
=
\enc_{\theta,\varphi}\,\overline g
\qquad(g\in\tA).
\]
The single-ququart physical representation $g\mapsto g$ and the logical representation $g\mapsto\overline g$ are faithful four-dimensional irreducible representations of the same group $\tA$. Viewed as two logical qubits, the codes admit transversal non-Clifford operations.
\end{theorem}

Each code encodes one four-level system into seven and corrects an arbitrary error on one physical four-level system. Identifying each ququart with a pair of qubits gives an encoding of two logical qubits into fourteen physical qubits, with transversal gates acting on seven disjoint pairs. This realizes the finite component of the super-golden construction; implementing the additional gates completing that universal set remains a separate problem.

The basic phenomenon exploited here is already visible in a familiar setting: we show that Steane's $[[7,1,3]]$ code \cite{steane}, as well as a nonstabilizer code with the same transversal Clifford and \(\GL(3,\mathbb F_2)\) permutation symmetries, arise from two points in a continuous family of invariant parent states. In general, for a local unitary $t$-group $\G\leq\SU(q)$ and a $t$-transitive permutation group $ \Sigma \leq S_n $, the same reasoning becomes an existence and construction theorem. Local design symmetry restricts each marginal on at most $t$ sites to the algebra of tensor-factor permutations; permutation symmetry makes it central. Schur--Weyl duality expresses maximal mixing as linear equations in the sector weights. These conditions are sufficient: every feasible weight vector produces pure parent states, with no remaining phase equations to solve. Puncturing gives $q$-dimensional codes of distance at least $t$ and transports the local symmetry to the logical system.

The remainder of the paper develops this mechanism in the Steane setting, proves the general construction theorem, and then applies it to \(\widetilde A_7\). An  example in the appendix, using the binary tetrahedral group, recovers the standard $ [[5,1,3]] $ code \cite{laflamme} and shows how additional time-reversal symmetry can raise the distance beyond the basic $t$-strength guarantee.

\section{The Steane code, the Fake Steane code, and an invariant family}
\label{sec:steane}

Let $H$ be the binary $[7,4,3]$ Hamming code. Its dual $H^\perp$ is the $[7,3,4]$ simplex code, with $H^\perp\subset H$ and every nonzero word of weight four. The $[[7,1,3]]$ Steane CSS code uses this simplex code for both its $X$- and $Z$-type stabilizers, giving transversal Hadamard (since the X and Z type stabilizer come from the same code) and transversal phase gates (since the simplex code is double even), and thus the entire single qubit Clifford group can be implemented transversally~\cite{steane}. Also, the seven physical qubits are indexed by the seven nonzero points of $\FF_2^3$, and thus the corresponding $\GL(3,\FF_2)$ permutation symmetry preserves the Steane code.

Surprisingly, these prescribed transversal Clifford and permutation actions fails to specify the code uniquely; they also occur for a nonstabilizer $((7,2,3))_2$ code, encountered by Eric Kubischta and the author in private correspondence in 2022 and called the Fake Steane code~\cite{kubischta_private_2022}. The comparison is explicit. Let $\mathcal S$ be the seven nonzero words of the simplex code $H^\perp$, and let $\mathcal T$ be the other $28$ weight-four bit strings of length seven. Both codes have $|1_\alpha\rangle=X^{\otimes7}|0_\alpha\rangle$, for $\alpha\in\{\mathrm{St},\mathrm F\}$, while
\[
\begin{aligned}
 |0_{\mathrm{St}}\rangle
 &=\frac1{\sqrt8}\left(|0000000\rangle+\sum_{x\in\mathcal S}|x\rangle\right),\\
 |0_{\mathrm F}\rangle
 &=\frac1{\sqrt{200}}\left(5|0000000\rangle-3\sum_{x\in\mathcal S}|x\rangle+2\sum_{x\in\mathcal T}|x\rangle\right).
\end{aligned}
\]
We write $|a_L\rangle=|a_{\mathrm{St}}\rangle$ for the usual Steane basis. 

Why do these different codes have the same transversal and permutation symmetries? We answer this question by showing that purification places both codes in a continuous family of eight-qubit invariant states, revealing what their common symmetries determine.

\subsection{Purification and the eight-site symmetry}

Let
\[
    W:\CC^2\longrightarrow(\CC^2)^{\otimes7}
\]
be an encoding isometry whose transversal \(2O\) action is contragredient to the defining two-dimensional representation:
\begin{equation}
\label{eq:seven-qubit-covariance}
    g^{\otimes7}W=W\overline g
    \qquad(g\in2O).
\end{equation}
Associate to \(W\) its normalized purification
\begin{equation}
\label{eq:parent-from-encoder}
    |\Omega_W\rangle
    =
    \frac1{\sqrt2}\sum_{a=0}^1W|a\rangle\otimes|a\rangle.
\end{equation}
Since
\[
    (\overline g\otimes g)
    \frac{|00\rangle+|11\rangle}{\sqrt2}
    =
    \frac{|00\rangle+|11\rangle}{\sqrt2},
\]
Equation~\eqref{eq:seven-qubit-covariance} implies
\[
    g^{\otimes8}|\Omega_W\rangle=|\Omega_W\rangle .
\]
Thus a \(2O\)-covariant seven-qubit code becomes a diagonally \(2O\)-invariant eight-qubit pure state.

There is a parallel enlargement of the permutation picture.  Label the
seven physical qubits by the nonzero points of $\FF_2^3$ and the purifying
qubit by $0$.  The original $\GL(3,\FF_2)$ permutation symmetry is then the
stabilizer of $0$ inside
\[
    \AGL(3,\FF_2)=\FF_2^3\rtimes\GL(3,\FF_2).
\]
For the Steane and Fake Steane codes, the purifications acquire the full
$3$-transitive permutation symmetry $\AGL(3,\FF_2)$ on all eight points. For Steane, the purification is the extended-Hamming state
\begin{equation}
\label{eq:hamming-parent}
    |\Omega_0\rangle
    =
    \frac14\sum_{c\in H_8}|c\rangle
    =
    \frac{|0_L\rangle|0\rangle+|1_L\rangle|1\rangle}{\sqrt2},
    \qquad H_8=[8,4,4]_2.
\end{equation}

To compare it directly with the Fake Steane parent, let $\cA$ be the
fourteen affine planes in $\FF_2^3$, let $\cN$ be the other fifty-six
four-element subsets, and write $|S\rangle$ for the binary indicator word
of $S\subseteq\FF_2^3$.  Since the words of $H_8$ are the empty set, the
whole set, and the affine planes,
\begin{equation}
\label{eq:hamming-parent-affine}
    |\Omega_0\rangle
    =
    \frac14
    \left(
        |00000000\rangle+|11111111\rangle
        +\sum_{S\in\cA}|S\rangle
    \right).
\end{equation}

The Fake Steane purification is the other real point
\begin{equation}
\label{eq:fake-steane-parent}
    |\Omega_\pi\rangle
    =
    \frac1{20}
    \left(
        5|00000000\rangle+5|11111111\rangle
        -3\sum_{S\in\cA}|S\rangle
        +2\sum_{S\in\cN}|S\rangle
    \right).
\end{equation}
It is likewise invariant under
$2O\times\AGL(3,\FF_2)$, and puncturing any coordinate gives the
Fake Steane code.

\subsection{The nonuniqueness is continuous}

The two real parents are only two points of a larger family.  For
\(\theta\in\RR/2\pi\ZZ\), define
\begin{equation}
\label{eq:steane-explicit}
\begin{aligned}
    |\Omega_\theta\rangle
    ={}&
    \frac{|00000000\rangle+|11111111\rangle}{4}
    +\frac{1+4e^{i\theta}}{20}
        \sum_{S\in\cA}|S\rangle \\
    &+
    \frac{1-e^{i\theta}}{20}
        \sum_{S\in\cN}|S\rangle.
\end{aligned}
\end{equation}
At \(\theta=0\), the non-plane terms disappear and
\(|\Omega_0 \rangle\) is the extended-Hamming
state~\eqref{eq:hamming-parent}.  At \(\theta=\pi\), it is the Fake Steane
parent~\eqref{eq:fake-steane-parent}.  Generic values of \(\theta\) give
properly complex states.

Every state in~\eqref{eq:steane-explicit} is normalized, invariant under
\(2O\times\AGL(3,\FF_2)\), and of exact pure distance four, as shown below.
Puncturing any coordinate gives a
\(((7,2,3))_2\) code with the same prescribed transversal \(2O\) action and
with the point stabilizer
\[
    \AGL(3,\FF_2)_0\cong\GL(3,\FF_2)
\]
acting by permutation symmetries on the 7 physical qubits.

Although the Steane and Fake Steane codes have the same ordinary
Shor--Laflamme weight enumerators,
\[
 A(z)=1+21z^4+42z^6,
 \qquad
 B(z)=1+21z^3+21z^4+126z^5+42z^6+45z^7,
\]
they are distinguished by a finer support-resolved weight enumerator.
For a code projector $P$ of rank $K$, define
\[
 A_\omega(P)
 =
 \frac1{K^2}
 \sum_{\operatorname{supp}(E)=\omega}
 |\Tr(EP)|^2,
\]
where the sum is over Hermitian Pauli operators with exact support
$\omega$; this is the ``frequency analysis'' refinement of the quantum
weight enumerator~\cite{yu_graphical}.  The quantities $A_\omega$ are
invariant under local unitaries.

Let $\mathcal S$ be the seven weight-four supports of the simplex code
and $\mathcal T$ the other twenty-eight weight-four supports.  Then
\[
 A_\omega^{\rm St}
 =
 \begin{cases}
 3,&\omega\in\mathcal S,\\
 0,&\omega\in\mathcal T,
 \end{cases}
 \qquad
 A_\omega^{\rm F}
 =
 \begin{cases}
 \frac{279}{125},&\omega\in\mathcal S,\\[1mm]
 \frac{24}{125},&\omega\in\mathcal T.
 \end{cases}
\]
All other support coefficients agree.  For a stabilizer code every
$A_\omega$ is an integer, so the fractional coefficients show that the
Fake Steane code is not locally-unitarily equivalent to any stabilizer
code.  The complete fine enumerators are recorded in
Appendix~\ref{app:weight-enumerators}.

\subsection{Two invariant components}

The reason the phase in~\eqref{eq:steane-explicit} is free becomes
transparent after decomposing the simultaneous invariant space.  Write
\[
    |\Omega_\theta\rangle=a+e^{i\theta}b,
\]
where
\[
    a=
    \frac{|00000000\rangle+|11111111\rangle}{4}
    +\frac1{20}\sum_{|S|=4}|S\rangle,
    \qquad
    b=
    \frac15\sum_{S\in\cA}|S\rangle
    -\frac1{20}\sum_{S\in\cN}|S\rangle.
\]
The vector \(a\) is fully symmetric, while \(b\) is an \(\SU(2)\) singlet.
Indeed, every three-element subset of \(\FF_2^3\) lies in one affine plane
and four non-planes, so the collective raising operator annihilates \(b\);
it also has spin weight zero.  Directly,
\[
    \|a\|^2=\frac3{10},
    \qquad
    \|b\|^2=\frac7{10},
    \qquad
    \langle a,b\rangle=0.
\]
Hence
\[
    e_8=\frac{a}{\sqrt{3/10}},
    \qquad
    e_{44}=\frac{b}{\sqrt{7/10}}
\]
span the two simultaneous invariant Schur--Weyl sectors
\([8]\) and \([4,4]\), and
\begin{equation}
\label{eq:steane-sectors}
    |\Omega_\theta\rangle
    =
    \sqrt{\frac3{10}}\,e_8
    +
    e^{i\theta}\sqrt{\frac7{10}}\,e_{44}.
\end{equation}
Both vectors are fixed by the permutation group \(\AGL(3,\FF_2)\).  Invariance
of \(a\) under the local Clifford group follows directly from
\(H^{\otimes8}\) and \(S^{\otimes8}\), or equivalently from the binary
invariant polynomial
\[
    x^8+14x^4y^4+y^8.
\]

The permutation moments explain why error correction fixes the component
norms but not their relative phase.  A normalized transposition class sum
acts on the sectors \([8]\) and \([4,4]\) by
\[
    1,\qquad \frac27,
\]
while a normalized three-cycle class sum acts by
\[
    1,\qquad -\frac1{14}.
\]
For maximally mixed two- and three-qubit marginals the corresponding target
moments are \(1/2\) and \(1/4\).  Therefore the sector weights satisfy
\[
    p_8+p_{44}=1,
    \qquad
    p_8+\frac27p_{44}=\frac12,
    \qquad
    p_8-\frac1{14}p_{44}=\frac14.
\]
The first two equations uniquely give
\[
    (p_8,p_{44})
    =
    \left(\frac3{10},\frac7{10}\right),
\]
and the third then holds automatically.  None of these equations sees the
relative phase.

The distance is exactly four throughout the circle.  A normalized
double-transposition class sum has expectation
\[
    \frac3{10}+\frac17\frac7{10}
    =
    \frac25,
\]
whereas a maximally mixed four-qubit marginal would require \(1/4\).
Thus no point of the circle has pure distance five.  By the puncturing
result proved in Section~\ref{sec:central-moment-method}, every one-site
puncture has exact distance three.

\section{When symmetry makes quantum code construction linear}
\label{sec:central-moment-method}

We now turn the phase freedom in Section~\ref{sec:steane} into a general construction. Local unitary-design symmetry and permutation transitivity reduce the maximal-mixing conditions to a linear feasibility problem in Schur--Weyl sector weights.

\subsection{Pure distance and permutation moments}

Let \(V=\CC^q\).  For a normalized state
\(\psi\in V^{\otimes n} \), write
\[
    \rho_S=\Tr_{S^c}|\psi\rangle\langle\psi|
\]
for the marginal on a set of sites \(S\).

We say that \(\psi\) has pure distance at least \(d\) when
\[
    \langle\psi|E|\psi\rangle=0
\]
for every traceless error \(E\) supported on fewer than \(d\) sites.  Equivalently,
\begin{equation}
\label{eq:mixing}
    \rho_S=\frac{I}{q^{|S|}}
    \qquad (|S|<d).
\end{equation}
Indeed, for an operator \(A\) supported on \(S\),
\[
    \langle\psi|A|\psi\rangle=\Tr(\rho_S A),
\]
so vanishing against every traceless \(A\) is equivalent to maximal mixing of
\(\rho_S\).  We write \(((n,1,d))_q\) for a one-dimensional state of exact pure distance \(d\).

The maximal-mixing conditions imply particularly simple expectations for tensor-factor permutations.  If
\(\sigma\in S_k\) has \(c(\sigma)\) cycles and is embedded on a \(k\)-set \(S\), then
\[
    \Tr_{V^{\otimes k}}(\sigma)=q^{c(\sigma)},
\]
and hence
\begin{equation}
\label{eq:local-permutation-moment}
    \langle\psi|\sigma_S|\psi\rangle
    =
    q^{c(\sigma)-k}
\end{equation}
whenever \(|S|=k<d\).

Let \(\mu\vdash k\) be a cycle type and let
\(C_{\mu,1^{n-k}}\subset S_n\) denote the corresponding conjugacy class with
\(n-k\) additional fixed points.  Define the normalized class sum
\[
    \Omega_{\mu,1^{n-k}}
    =
    \frac1{|C_{\mu,1^{n-k}}|}
    \sum_{\pi\in C_{\mu,1^{n-k}}}\pi .
\]
Averaging~\eqref{eq:local-permutation-moment} over the support gives the
central moment equations
\begin{equation}
\label{eq:central-moment}
    \langle\psi|\Omega_{\mu,1^{n-k}}|\psi\rangle
    =
    q^{c(\mu)-k}.
\end{equation}
These conditions are always necessary for pure distance \(>k\).  Without
additional symmetry they need not be sufficient.

\subsection{Schur--Weyl linearization}

Schur--Weyl duality gives
\[
    V^{\otimes n}
    \cong
    \bigoplus_{\lambda\vdash n,\ \ell(\lambda)\leq q}
    M_\lambda\otimes[\lambda],
\]
where \(M_\lambda\) is the irreducible \(\SU(q)\) representation of highest weight
\(\lambda\), and \([\lambda]\) is the Specht module
of \(S_n\).

Let a finite local group \(\G\leq\SU(q)\) act diagonally and let a permutation
group \(\Sigma\leq S_n\) act by permuting tensor factors.  Their simultaneous
invariant space decomposes as
\[
    (V^{\otimes n})^{\G\times\Sigma}
    \cong
    \bigoplus_\lambda W_\lambda,
    \qquad
    W_\lambda
    :=
    M_\lambda^\G\otimes[\lambda]^\Sigma.
\]

For a normalized invariant state \(\psi\), let \(P_\lambda\) denote the projector onto the full Schur--Weyl sector \(M_\lambda\otimes[\lambda]\), and write
\[
    p_\lambda=\|P_\lambda\psi\|^2.
\]
Then
\[
    p_\lambda\geq0,
    \qquad
    \sum_\lambda p_\lambda=1,
    \qquad
    p_\lambda=0\ \text{if }W_\lambda=0.
\]

For any conjugacy class $C\subseteq S_n$, write
\[
    \Omega_C=\frac1{|C|}\sum_{\pi\in C}\pi
\]
for its normalized class sum.  Since $\Omega_C$ is central in
$\CC[S_n]$, it acts on $[\lambda]$ by the character ratio
\[
    \omega_\lambda(C)
    =
    \frac{\chi^\lambda(C)}{\dim[\lambda]}.
\]
Hence
\begin{equation}
\label{eq:linear-moment}
    \langle\psi|\Omega_C|\psi\rangle
    =
    \sum_\lambda
    p_\lambda
    \frac{\chi^\lambda(C)}{\dim[\lambda]}.
\end{equation}

The central moment conditions therefore become the linear feasibility problem
\begin{equation}
\label{eq:lp}
\begin{gathered}
    p_\lambda\geq0,\qquad
    \sum_\lambda p_\lambda=1,\qquad
    p_\lambda=0\ \text{if }W_\lambda=0,\\[2mm]
    \sum_\lambda
    p_\lambda
    \frac{\chi^\lambda(\mu,1^{n-k})}{\dim[\lambda]}
    =
    q^{c(\mu)-k}
    \qquad
    (1\leq k\leq t,\ \mu\vdash k).
\end{gathered}
\end{equation}

Thus the original quadratic conditions on amplitudes have disappeared from
the central part of the problem.  The unknowns in~\eqref{eq:lp} are only the
sector probabilities \(p_\lambda\).

\subsection{Completeness under local design and permutation symmetry}

\begin{definition}
A finite group \(\G\leq\SU(q)\) is a unitary \(t\)-group if
\[
    \End_\G(V^{\otimes r})
    =
    \End_{\SU(q)}(V^{\otimes r})
    \qquad
    (1\leq r\leq t).
\]
Equivalently, if it is a unitary \(t\)-design.
\end{definition}

The next theorem is the central construction principle.

\begin{theorem}[Linear construction of invariant parents]
\label{thm:central-moments-complete}
Let \(\G\leq\SU(q)\) be a unitary \(t\)-group and let
\(\Sigma\leq S_n\) be a \(t\)-transitive permutation group.  There exists a
normalized \(\G\times\Sigma\)-invariant state whose marginals on all sets of at
most \(t\) sites are maximally mixed if and only if~\eqref{eq:lp} is feasible.

More precisely, for every feasible collection \((p_\lambda)\) and every choice
of unit vectors
\[
    u_\lambda\in W_\lambda
    \qquad (p_\lambda>0),
\]
the state
\begin{equation}
\label{eq:constructed-parent}
    \psi
    =
    \sum_{\lambda:p_\lambda>0}
    \sqrt{p_\lambda}\,u_\lambda
\end{equation}
has pure distance at least \(t+1\).
\end{theorem}

\begin{proof}
Necessity follows from~\eqref{eq:central-moment} and
\eqref{eq:linear-moment}.  Conversely, let \((p_\lambda)\) be feasible and
construct \(\psi\) as in~\eqref{eq:constructed-parent}.  Fix a set
\(S\) of \(r\leq t\) sites. Diagonal \(\G\)-invariance implies that \(\rho_S\) commutes with \(\G\) on
\(V^{\otimes r}\).  Since \(\G\) is a unitary \(t\)-group,
\[
    \rho_S
    \in
    \End_\G(V^{\otimes r})
    =
    \operatorname{im}\CC[S_r].
\]
Because \(\Sigma\) is \(r\)-transitive, the setwise stabilizer of \(S\)
induces the full symmetric group \(S_r\) on \(S\).  Permutation invariance
therefore implies
\[
    \rho_S
    \in
    Z\!\left(\operatorname{im}\CC[S_r]\right).
\]

Write
\[
    \rho_S=\sum_\nu a_{\nu,S}\Pi_{\nu,S}
\]
using the local Schur--Weyl central idempotents.  For every cycle type
\(\mu\vdash r\), all embeddings of a permutation of type \(\mu\) on \(S\)
have the same expectation by \(r\)-transitivity.  Thus the global moment in
\eqref{eq:lp} equals the local class-sum moment:
\[
    \Tr(\rho_S\Omega_\mu)
    =
    q^{c(\mu)-r}
    =
    \Tr\!\left(\frac{I}{q^r}\Omega_\mu\right).
\]
The class sums span the center of the image of \(\CC[S_r]\), equivalently
their character ratios separate its central idempotents.  Hence all
coefficients \(a_{\nu,S}\) are forced to equal those of \(I/q^r\), and
\[
    \rho_S=\frac{I}{q^r}.
\]
This holds for every \(r\leq t\).
\end{proof}

The proof can be summarized in three steps:
\[
\boxed{
\begin{aligned}
\G\text{ a unitary }t\text{-group}
&\Longrightarrow
\rho_S\in\operatorname{im}\CC[S_r],\\
\Sigma\text{ }t\text{-transitive}
&\Longrightarrow
\rho_S\in Z\!\left(\operatorname{im}\CC[S_r]\right),\\
\text{central moments feasible}
&\Longrightarrow
\rho_S=\frac{I}{q^r}.
\end{aligned}}
\]
The local design property determines the algebra in which the marginal can live, while permutation transitivity makes it central.

For a fixed feasible solution, the certified family of parents is
\[
    \left(
    \prod_{\lambda:p_\lambda>0}
    \mathbb S(W_\lambda)
    \right)\big/U(1),
\]
where \(\mathbb S(W)\) denotes the unit sphere in \(W\).  In particular, if
\(s\) contributing invariant sectors are one-dimensional, then the remaining
freedom is a torus \(T^{s-1}\) of relative phases.  This is exactly what
happens for the Steane circle and the \(\tA\) two-torus.

\subsection{Purification and puncturing}

The seven-qubit discussion in Section~\ref{sec:steane} can be read in either
direction.  A covariant code can be purified to an invariant parent, while an
invariant parent with a maximally mixed one-site marginal can be punctured to
a covariant code.

\begin{proposition}[Purification of a covariant code]
\label{prop:purification}
Let
\[
    W:V\longrightarrow V^{\otimes(n-1)}
\]
be an isometry satisfying
\[
    g^{\otimes(n-1)}W=W\overline g
    \qquad(g\in\G).
\]
Then
\[
    |\Omega_W\rangle
    =
    \frac1{\sqrt q}
    \sum_{a=0}^{q-1}W|a\rangle\otimes|a\rangle
\]
is invariant under the diagonal action \(g^{\otimes n}\).
\end{proposition}

\begin{proof}
The maximally entangled vector
\[
    |\Phi_q\rangle
    =
    \frac1{\sqrt q}\sum_a|a\rangle|a\rangle
\]
satisfies
\[
    (\overline g\otimes g)|\Phi_q\rangle=|\Phi_q\rangle.
\]
Since
\[
    |\Omega_W\rangle=(W\otimes I)|\Phi_q\rangle,
\]
the covariance relation gives
\[
    g^{\otimes n}|\Omega_W\rangle=|\Omega_W\rangle.
\]
\end{proof}

\begin{corollary}[Codes from invariant parents]
\label{cor:puncture}
Let \(\psi\in V^{\otimes n}\) be a diagonally \(\G\)-invariant pure state
whose marginals on all sets of at most \(t\) sites are maximally mixed.
Puncturing any one site produces a \(q\)-dimensional code
\(\cC\subset V^{\otimes(n-1)}\) with distance at least \(t\).  The encoding
may be chosen so that
\begin{equation}
\label{eq:covariance}
    g^{\otimes(n-1)}\enc
    =
    \enc\,\overline g
    \qquad(g\in\G).
\end{equation}
\end{corollary}

\begin{proof}
For the deleted site \(r\), write the Schmidt decomposition
\[
    |\psi\rangle
    =
    \frac1{\sqrt q}
    \sum_{a=0}^{q-1}|\phi_a\rangle|a\rangle,
    \qquad
    |\phi_a\rangle
    =
    \sqrt q\,
    (I\otimes\langle a|)|\psi\rangle.
\]
Because the one-site marginal is \(I/q\), the \(\phi_a\) are orthonormal.

Let \(E\) be supported on a set \(S\) of fewer than \(t\) remaining sites.
Then
\[
\begin{aligned}
    \langle\phi_a|E|\phi_b\rangle
    &=
    q\,
    \langle\psi|
    E\otimes|a\rangle\langle b|_r
    |\psi\rangle\\
    &=
    \frac{\Tr(E)}{q^{|S|}}\delta_{ab},
\end{aligned}
\]
because the marginal on \(S\cup\{r\}\) is maximally mixed.  These are the
Knill--Laflamme conditions, so the punctured code has distance at least
\(t\).

Finally, diagonal \(\G\)-invariance implies
\[
    g^{\otimes(n-1)}|\phi_a\rangle
    =
    \sum_b\overline{g_{ba}}|\phi_b\rangle,
\]
which is~\eqref{eq:covariance}.
\end{proof}

For a unitary \(t\)-group together with a \(t\)-transitive permutation group,
a feasible solution gives a punctured code of distance at least $ t $.  Additional
symmetry can sometimes raise the distance beyond the basic $ t $ strength guarantee;
Appendix~\ref{app:tetrahedral} gives a unitary-\(2\)-group example in which
time-reversal symmetry supplies the missing condition and forces some of the punctured codes to have distance $ 3 $ rather than distance $ 2 $.

\section{Exceptional transversal gates: the \texorpdfstring{$\tA$}{A7} family}\label{sec:a7-main}
Now replace the local binary octahedral symmetry $ 2O \leq \SU(2) $ by $\tA\leq\SU(4)$ and keep the same permutation symmetry group $\AGL(3,\FF_2)$ on the eight points of $\FF_2^3$. This four-dimensional representation is a unitary $3$-group~\cite{cliffordsubgroups,bannai}; Appendix~\ref{app:certificates} also gives a finite character-moment certificate. Thus the construction theorem applies with $q=4$, $n=8$, and $t=3$.

\begin{table}[tb]
\centering
\renewcommand{\arraystretch}{1.12}
\caption{}\label{tab:steane-a7}
\begin{tabular}{@{}lll@{}}\toprule
&Steane family&$\tA$ family\\\midrule
Local dimension&$2$&$4$\\
Local symmetry&$2O$&$\tA$\\
Permutation symmetry&$\AGL(3,\FF_2)$&$\AGL(3,\FF_2)$\\
Invariant sectors&$[8],[4,4]$&$[8],[4,4],[2,2,2,2]$\\
Parent parameters&$S^1$&$T^2$\\
Pure parent distance&$4$&$4$\\
Punctured code&$((7,2,3))_2$&$((7,4,3))_4$\\\bottomrule
\end{tabular}
\end{table}

\subsection{Three sectors and three equations}
For a Schur sector $M_\lambda\otimes[\lambda]$, the multiplicity of simultaneous invariants is
\[
 \dim M_\lambda^{\tA}\,\dim[\lambda]^{\AGL(3,\FF_2)}.
\]
The character averages in Appendix~\ref{app:certificates} leave precisely three one-dimensional invariant sectors,
\[
 [2,2,2,2],\qquad[4,4],\qquad[8].
\]
The first is an $\SU(4)$ singlet; the other two are finite-group invariant sectors. Consequently
\[
 \bigl((\CC^4)^{\otimes8}\bigr)^{\tA\times\AGL(3,\FF_2)}
 =\CC e_{2222}\oplus\CC e_{44}\oplus\CC e_8,
\]
with normalized vectors given in Appendix~\ref{app:a7-explicit}. The relevant character ratios are
\[
\begin{array}{c|cc}
\lambda&\omega_\lambda(2,1^6)&\omega_\lambda(3,1^5)\\\hline
{[2,2,2,2]}&-2/7&-1/14\\
{[4,4]}&2/7&-1/14\\
{[8]}&1&1
\end{array}
\]
and the maximal-mixing targets are $1/4$ and $1/16$. Writing $p_{2222},p_{44},p_8$ for the sector weights, we obtain
\begin{equation}\label{eq:a7-linear}
\begin{aligned}
 p_{2222}+p_{44}+p_8&=1,\\
 -\frac27p_{2222}+\frac27p_{44}+p_8&=\frac14,\\
 -\frac1{14}(p_{2222}+p_{44})+p_8&=\frac1{16}.
\end{aligned}
\end{equation}
The unique solution is
\[
 (p_{2222},p_{44},p_8)=\left(\frac7{32},\frac{21}{32},\frac18\right).
\]
Only the norms have been fixed. Every choice of the two relative phases therefore gives a parent state
\begin{equation}\label{eq:a7-torus}
 |\Psi_{\theta,\varphi}\rangle
 =\sqrt{\frac7{32}}e_{2222}
 +e^{i\theta}\sqrt{\frac{21}{32}}e_{44}
 +e^{i\varphi}\sqrt{\frac18}e_8
\end{equation}
of pure distance at least four. For a double transposition the character ratios are $1/7,1/7,1$, so
\[
 \langle\Omega_{(2,2)}\rangle
 =\frac17\left(\frac7{32}+\frac{21}{32}\right)+\frac18=\frac14.
\]
Pure distance at least five would require $4^{2-4}=1/16$. Every point of the two-torus thus has exact pure distance four.

\subsection{A single-error-correcting code}
Puncturing any site gives a four-dimensional codespace with logical basis
\[
 |a_L\rangle=2\bigl(I^{\otimes7}\otimes\langle a|\bigr)|\Psi_{\theta,\varphi}\rangle,
 \qquad a=0,1,2,3.
\]
Corollary~\ref{cor:puncture} gives exact parameters $((7,4,3))_4$ and
\[
g^{\otimes7}\enc_{\theta,\varphi}=\enc_{\theta,\varphi}\overline g
 \qquad(g\in\tA).
\]
The logical action is the complex conjugate of the defining physical action. Note that $\tA$ contains elements of order seven. These cannot be two-qubit Clifford gates: the two-qubit Clifford group has projective order $11520$, which is not divisible by seven~\cite{cliffordsubgroups}. This gives a concrete witness to the non-Clifford transversal action and completes the proof of Theorem~\ref{thm:main}.

\section*{Acknowledgments}
The author would like to thank Eric Kubischta for many years of collaboration on the topic of quantum error-correcting codes with exotic transversal gates, and in particular for his insight on the relationship between invariant states and covariant codes, the Fake Steane code, and sundry other topics. Computations for this project were completed using OpenAI ChatGPT and the group theory software GAP (Groups Algorithms and Programming).

\appendices
\section{Characters and invariant-space dimensions}\label{app:certificates}
\subsection{Symmetric-group coefficients}
The normalized class sum acts on $[\lambda]$ by $\omega_\lambda(\mu)=\chi^\lambda(\mu)/\dim[\lambda]$. For a transposition, the content formula gives
\[
 \omega_\lambda(2,1^{n-2})
 =\binom n2^{-1}\sum_{(r,s)\in\lambda}(s-r).
\]
The hook-length formula gives $\dim[\lambda]$, and the Murnaghan--Nakayama rule gives the remaining character values~\cite{sagan}. For $[4,4]$, the content sum is $8$ and the dimension is $14$; the three-cycle character is $-1$. Its conjugate partition $[2,2,2,2]$ has the opposite transposition ratio and the same ratios on even permutations. This gives the tables in the main text, including the double-transposition ratios $1/7,1/7,1$.

\subsection{The permutation symmetry group}
The affine maps $x\mapsto Ax+b$, $A\in\GL(3,\FF_2)$, $b\in\FF_2^3$, give the $1344$ elements of $\Sigma=\AGL(3,\FF_2)$. Their cycle counts are
\[
\begin{array}{c|rrrrrrrr}
\text{type}&1^8&2^2 1^4&2^4&3^2 1^2&4^1 2^1 1^2&4^2&6^1 2^1&7^1 1^1\\\hline
\text{number}&1&42&49&224&168&252&224&384.
\end{array}
\]
Thus $m_\Sigma(\lambda)=1344^{-1}\sum_{\sigma\in\Sigma}\chi^\lambda(\sigma)$ is directly evaluable. For the local group, write $s_\lambda(g)$ for the Schur character of $M_\lambda$ and use
\begin{equation}\label{eq:invariant-averages}
 m_\G(\lambda)=\frac1{|\G|}\sum_{g\in\G}s_\lambda(g),
 \qquad \dim W_\lambda=m_\G(\lambda)m_\Sigma(\lambda).
\end{equation}
For all partitions of eight with at most four rows these evaluations give
\[
\begin{array}{c|cc}
\lambda&m_\Sigma(\lambda)&m_{\tA}(\lambda)\\\hline
{[8]}&1&1\\
{[4,4]}&1&1\\
{[2,2,2,2]}&1&1\\
\text{all other such partitions}&0&0.
\end{array}
\]
For $q=2$ only partitions with at most two rows occur; the local $2O$ averages on $[8]$ and $[4,4]$ are both one. This proves the invariant-space dimensions used for Steane and $A_7$.

\section{Explicit data for the \texorpdfstring{$\tA$}{A7} invariant space}\label{app:a7-explicit}
\subsection{The chosen lift}
Put $s=(1+\sqrt{-7})/2$, $\bar s=(1-\sqrt{-7})/2$, and $\zeta_7=e^{2\pi i/7}$. We use $\tA=\langle V_1,V_2,D\rangle\leq\SU(4)$, where
\[
 V_1=\begin{pmatrix}
 1&0&0&0\\0&0&0&\zeta_7^{-1}\\0&\zeta_7^5&0&0\\0&0&\zeta_7^3&0
 \end{pmatrix},\quad
 V_2=\frac1{\sqrt7}\begin{pmatrix}
 \bar s^2&1&1&1\\1&s&\bar s&\bar s\\1&\bar s&s&\bar s\\1&\bar s&\bar s&s
 \end{pmatrix},\quad
 D=\begin{pmatrix}
 1&0&0&0\\0&0&0&1\\0&1&0&0\\0&0&1&0
 \end{pmatrix}.
\]
The matrices $V_1,-iV_2,D$ generate the four-dimensional representation of the double cover of $A_7$ in Ref.~\cite[Appendix~B.2]{cliffordsubgroups}. Since $V_2^3=-iI$, our group is its scalar extension by $iI$ and has order $10080$ and scalar center $\langle iI\rangle$. Its projective image is $A_7$. On eight tensor factors these scalar factors act trivially, so the two choices have the same invariant parents and operator commutants. 

\subsection{Orbit notation}

Identify the eight ququart sites with \(\FF_2^3\), ordered as
\[
    000,\ 100,\ 010,\ 110,\ 001,\ 101,\ 011,\ 111.
\]
For an \(\AGL(3,\FF_2)\)-orbit \(\mathcal O\) of computational-basis words
\(w\in\{0,1,2,3\}^8\), write
\[
    |\mathcal O\rangle
    =
    |\mathcal O|^{-1/2}
    \sum_{w\in\mathcal O}|w\rangle.
\]
Orbit labels are written as representative words in the site order above.

\subsection{The three invariant vectors}

The following normalized vectors span the three invariant sectors certified in Appendix~\ref{app:certificates}.

The \(\SU(4)\)-singlet vector is
\[
\begin{aligned}
e_{2222}
={}&
\frac{1}{\sqrt3}|\mathcal O_{33221100}\rangle
-\frac{1}{\sqrt6}|\mathcal O_{32321100}\rangle
-\frac{1}{\sqrt6}|\mathcal O_{33212100}\rangle \\
&+
\frac{1}{\sqrt6}|\mathcal O_{32312100}\rangle
-\frac{1}{\sqrt6}|\mathcal O_{22313100}\rangle .
\end{aligned}
\]

The \([4,4]\)-sector vector is
\[
    e_{44}
    =
    \sum_{\mathcal O}c_{\mathcal O}|\mathcal O\rangle,
\]
with nonzero coefficients listed after multiplication by $210$:
\[
\begin{array}{@{}crl@{\qquad}crl@{}}
\text{orbit} & |\mathcal O| & 210c_{\mathcal O} & \text{orbit} & |\mathcal O| & 210c_{\mathcal O}\\
\hline
31110000  &  56    & 8 \sqrt{21} &  32312100  &  1344  & 9 \sqrt{14}\\[1pt]
22210000  &  56    & 8 \sqrt{21} &  22313100  &  336   & - 3 \sqrt{14}\\[1pt]
33320000  &  56    & 8 \sqrt{21} &  33211110  &  672   & - 12 \sqrt{7}\\[1pt]
31101000  &  224   & - 4 \sqrt{21} &  32222110  &  672   & - 12 \sqrt{7}\\[1pt]
22201000  &  224   & - 4 \sqrt{21} &  33323210  &  672   & - 12 \sqrt{7}\\[1pt]
33302000  &  224   & - 4 \sqrt{21} &  13312110  &  168   & 24 \sqrt{7}\\[1pt]
22111000  &  336   & - 12 \sqrt{14} &  22223110  &  168   & 24 \sqrt{7}\\[1pt]
33222000  &  336   & - 12 \sqrt{14} &  33332210  &  168   & 24 \sqrt{7}\\[1pt]
33113000  &  336   & - 12 \sqrt{14} &  32221111  &  56    & - 8 \sqrt{21}\\[1pt]
33311000  &  224   & 12 \sqrt{21} &  33332111  &  56    & - 8 \sqrt{21}\\[1pt]
21112000  &  224   & 12 \sqrt{21} &  23323221  &  56    & - 8 \sqrt{21}\\[1pt]
32223000  &  224   & 12 \sqrt{21} &  32212111  &  224   & 4 \sqrt{21}\\[1pt]
33221100  &  168   & - 18 \sqrt{7} &  33323111  &  224   & 4 \sqrt{21}\\[1pt]
32321100  &  336   & - 3 \sqrt{14} &  33322221  &  224   & 4 \sqrt{21}\\[1pt]
33212100  &  336   & - 3 \sqrt{14} &    &   &  \\[1pt]
\end{array}
\]
and all omitted orbit coefficients vanish.

Finally, \(e_8\) is the degree-eight \(\tA\)-invariant polynomial, written in the normalized occupation basis. For $a_0+a_1+a_2+a_3=8$, define
\[
 |a_0,a_1,a_2,a_3\rangle
 =\binom{8}{a_0,a_1,a_2,a_3}^{-1/2}
 \sum_{\substack{w\in\{0,1,2,3\}^8\\n_r(w)=a_r\ (r=0,1,2,3)}}|w\rangle,
\]
where $n_r(w)$ counts the occurrences of the symbol $r$ in $w$.
The normalized symmetric invariant is
\[
\begin{aligned}
e_8={}&
\frac{\sqrt2}{7}
\bigl(
|0,0,2,6\rangle+
|0,2,6,0\rangle+
|0,6,0,2\rangle
\bigr)
\\
&-
\frac{4\sqrt5}{35}
\bigl(
|0,1,4,3\rangle+
|0,3,1,4\rangle+
|0,4,3,1\rangle
\bigr)
\\
&-
\frac{3\sqrt7}{49}
\bigl(
|1,0,0,7\rangle+
|1,0,7,0\rangle+
|1,7,0,0\rangle
\bigr)
\\
&-
\frac{\sqrt{15}}{35}
\bigl(
|1,1,2,4\rangle+
|1,2,4,1\rangle+
|1,4,1,2\rangle
\bigr)
\\
&+
\frac{\sqrt3}{7}
\bigl(
|2,0,5,1\rangle+
|2,1,0,5\rangle+
|2,5,1,0\rangle
\bigr)
+
\frac{\sqrt5}{35}|2,2,2,2\rangle
\\
&+
\frac{3\sqrt{10}}{35}
\bigl(
|3,0,3,2\rangle+
|3,2,0,3\rangle+
|3,3,2,0\rangle
\bigr)
\\
&-
\frac{\sqrt5}{35}
\bigl(
|4,0,1,3\rangle+
|4,1,3,0\rangle+
|4,3,0,1\rangle
\bigr)
\\
&+
\frac{\sqrt6}{7}|5,1,1,1\rangle
+
\frac{2\sqrt{14}}{49}|8,0,0,0\rangle .
\end{aligned}
\]
The vectors \(e_{2222},e_{44},e_8\) are normalized and mutually
orthogonal.

\section{Fine weight enumerators and the Fake Steane code}
\label{app:weight-enumerators}

The ordinary Shor--Laflamme weight enumerator does not distinguish
Steane from Fake Steane: both codes have
\[
    A(z)=1+21z^4+42z^6.
\]
A support-resolved refinement does.

Let $P$ be the projector onto a $K$-dimensional $n$-qubit code.
For a subset $\omega\subseteq\{1,\ldots,n\}$ define
\begin{equation}
\label{eq:fine-A}
    A_\omega(P)
    =
    \frac1{K^2}
    \sum_{\substack{E\in\{I,X,Y,Z\}^{\otimes n}\\
                    \supp(E)=\omega}}
    |\Tr(EP)|^2.
\end{equation}
These quantities refine the ordinary coefficient
\[
    A_j=\sum_{|\omega|=j}A_\omega
\]
and are invariant under local unitaries; this is the support-resolved
``frequency analysis'' used in~\cite{yu_graphical}.

Let $\mathcal S$ be the seven weight-four supports of the nonzero
words of the $[7,3,4]_2$ simplex code, and let $\mathcal T$ be the
other twenty-eight four-subsets.  Direct calculation gives
\[
\begin{array}{c|cc}
 &A_\omega^{\rm St}&A_\omega^{\rm F}\\
\hline
\omega\in\mathcal S
 &3&279/125\\[1mm]
\omega\in\mathcal T
 &0&24/125 .
\end{array}
\]
For both codes the remaining nonzero coefficients are
\[
    A_\varnothing=1,
    \qquad
    A_\omega=6\quad(|\omega|=6).
\]
Consequently,
\[
    7\cdot 3
    =
    7\cdot\frac{279}{125}
    +28\cdot\frac{24}{125}
    =
    21,
\]
so summing over supports reproduces the same ordinary weight
enumerator for both codes.

For a stabilizer code, every $A_\omega$ is an integer: it counts
stabilizer elements with exact support $\omega$.  The fractional
coefficients above therefore rule out stabilizer structure for the
Fake Steane code.  Since $A_\omega$ is invariant under local unitaries,
they prove the stronger statement: the Fake Steane code is not locally-unitarily equivalent
to any stabilizer code.

The two eight-qubit parent states behave in exactly the
same way.  Their ordinary enumerators agree,
\[
    A_{\Omega_0}(z)=A_{\Omega_\pi}(z)
    =1+42z^4+168z^6+45z^8,
\]
while on weight-four supports
\[
\begin{array}{c|cc}
 &A_\omega^{\Omega_0}&A_\omega^{\Omega_\pi}\\
\hline
\omega\text{ an affine plane}
 &3&279/125\\[1mm]
\omega\text{ a non-plane}
 &0&24/125 .
\end{array}
\]
Thus the Fake Steane parent is likewise not locally-unitarily
equivalent to a stabilizer state.

\section{The tetrahedral family, $ [[5,1,3]] $ code, and time reversal}
\label{app:tetrahedral}

\subsection{The invariant circle}

Following the gate convention of~\cite{exotic}, let
\[
    \mathsf Z=-iZ,
    \qquad
    \mathsf F=e^{-i\pi/4}HS^\dagger
    =
    \frac{e^{-i\pi/4}}{\sqrt2}
    \begin{pmatrix}
        1&-i\\
        1&i
    \end{pmatrix}.
\]
The facet gate $\mathsf F$ cyclically permutes the Pauli axes,
\[
    X\longmapsto Y,\qquad
    Y\longmapsto Z,\qquad
    Z\longmapsto X,
\]
and
\[
    2T=\langle \mathsf Z,\mathsf F\rangle\leq\SU(2)
\]
is the binary tetrahedral group.

Let $\PSL(2,\FF_5) $ act on the six coordinates
\[
    (\infty,0,1,2,3,4).
\]
Its weight-three words split into two orbits
\[
\begin{aligned}
 \mathcal O_{3a}={}&
 \{000111,001110,010011,011001,011100,\\
 &\hspace{14mm}100101,101001,101010,110010,110100\},\\
 \mathcal O_{3b}={}&
 \{001011,001101,010101,010110,011010,\\
 &\hspace{14mm}100011,100110,101100,110001,111000\}.
\end{aligned}
\]
Let $\mathcal O_1$ and $\mathcal O_5$ denote the words of weights one
and five.  A real orthonormal basis of the simultaneous invariant
space is
\[
 |T\rangle
 =
 \frac{-\sum_{\mathcal O_1}|x\rangle
       +\sum_{\mathcal O_5}|x\rangle}{\sqrt{12}},
 \qquad
 |S\rangle
 =
 \frac{-\sum_{\mathcal O_{3a}}|x\rangle
       +\sum_{\mathcal O_{3b}}|x\rangle}{\sqrt{20}}.
\]
The vector $|T\rangle$ lies in the $[6]$ sector and $|S\rangle$ in
the singlet sector $[3,3]$.

The transposition moment gives
\[
 p_{[6]}+p_{[3,3]}=1,
 \qquad
 p_{[6]}+\frac15p_{[3,3]}=\frac12,
\]
hence
\[
 (p_{[6]},p_{[3,3]})
 =
 \left(\frac38,\frac58\right).
\]
Therefore every state
\begin{equation}
\label{eq:tet-circle}
 |\psi_\theta\rangle
 =
 \sqrt{\frac38}|T\rangle
 +
 e^{i\theta}\sqrt{\frac58}|S\rangle
\end{equation}
has pure distance at least three.

\subsection{The two real points}

Let $\kappa$ denote complex conjugation in the computational basis and
define time reversal by
\[
    \Theta_6=(iY)^{\otimes6}\kappa.
\]
For a Hermitian Pauli string $E$,
\[
    \Theta_6E\Theta_6^{-1}
    =
    (-1)^{\wt(E)}E.
\]
Thus a time-reversal-invariant state has vanishing expectation against
every odd-weight Hermitian Pauli operator~\cite{timereversal}.

The entire circle~\eqref{eq:tet-circle} is invariant under the diagonal
$2T$ action, and in particular under $(iY)^{\otimes6}$.  At
$\theta=0,\pi$ the state is also real, so
\[
    \Theta_6|\psi_\theta\rangle=|\psi_\theta\rangle.
\]
The unitary-$2$-group construction already removes weights one and two,
while time reversal removes weight three.  Hence the two real points
have pure distance four.  This is the maximum possible pure distance
for six qubits, so both are exact distance-four states.

Away from the real locus the enhancement disappears.  Direct evaluation
of the first unconstrained Shor--Laflamme coefficient gives
\[
    A_3(\psi_\theta)
    =
    \frac{15}{2}\sin^2\theta.
\]
Thus the two real points are exactly the distance-four points of the
circle; all other points have exact pure distance three.

\subsection{Stabilizer parents and five-qubit punctures}

The two real points are stabilizer states. 
The positive real point is especially transparent.  A convenient
stabilizer generating set is
\[
\begin{aligned}
&IXZZXI,\qquad IIXZZX,\qquad IXIXZZ,\qquad IZXIXZ,\\
&-XXXXXX,\qquad -ZZZZZZ.
\end{aligned}
\]
Taking the first coordinate as the purifying qubit, puncturing it leaves
\[
    XZZXI,\qquad IXZZX,\qquad XIXZZ,\qquad ZXIXZ,
\]
the standard cyclic stabilizer presentation of the
$[[5,1,3]]$ perfect code~\cite{laflamme}.

For the negative real point $\psi_\pi$, one may similarly choose
\[
\begin{aligned}
&IYZZYI,\qquad IIYZZY,\qquad IYIYZZ,\qquad IZYIYZ,\\
&-YYYYYY,\qquad -ZZZZZZ.
\end{aligned}
\]
Puncturing the first coordinate gives
\[
    YZZYI,\qquad IYZZY,\qquad YIYZZ,\qquad ZYIYZ.
\]
This is the image of the standard five-qubit stabilizer under
$S^{\otimes5}$, since $SXS^\dagger=Y$ and $SZS^\dagger=Z$.
Hence the two real points puncture to two Clifford-equivalent
presentations of the five-qubit perfect code.

\end{document}